\documentclass[onecolumn,amsmath,amssymb,10pt,aps]{revtex4}
  
\usepackage[utf8]{inputenc}
\usepackage[T1]{fontenc}

\usepackage{amsmath}
\usepackage{amsfonts}
\usepackage{graphicx}
\usepackage{yfonts}
\usepackage{color}
\usepackage[normalem]{ulem}
\usepackage{amsthm}
\usepackage{bm}
\usepackage{bbm}
\usepackage{mathtools}
\usepackage{array}
\usepackage{placeins}
\usepackage{enumitem}
\def\oper{{\mathchoice{\rm 1\mskip-4mu l}{\rm 1\mskip-4mu l}
{\rm 1\mskip-4.5mu l}{\rm 1\mskip-5mu l}}}

\def\<{\langle}
\def\>{\rangle}

\newcommand{\Tr}{\mathrm{Tr}}
\def\oper{{\mathchoice{\rm 1\mskip-4mu l}{\rm 1\mskip-4mu l}
{\rm 1\mskip-4.5mu l}{\rm 1\mskip-5mu l}}}
\DeclareMathAlphabet\mathbfcal{OMS}{cmsy}{b}{n}

\mathchardef\mhyphen="2D 

\newcommand{\BH}{\mathcal{B}(\mathcal{H})}
\newcommand{\BHH}{\mathcal{B}_{\rm H}(\mathcal{H})}
\newcommand{\HH}{\mathcal{H}}
\newtheorem{Theorem}{Theorem}
\newtheorem{Lemma}{Lemma}
\newtheorem{Corollary}{Corollary}
\newtheorem{Definition}{Definition}

\newtheorem{Proposition}{Proposition}
\newtheorem{Example}{Example}

\begin{document}

\title{Interpolating between positive and completely positive maps: a new hierarchy of entangled states}

\author{Katarzyna Siudzi\'{n}ska, Sagnik Chakraborty, and Dariusz Chru\'{s}ci\'{n}ski}
\affiliation{Institute of Physics, Faculty of Physics, Astronomy and Informatics \\  Nicolaus Copernicus University, ul. Grudzi\k{a}dzka 5/7, 87--100 Toru\'{n}, Poland}

\begin{abstract}
A new class of positive maps is introduced. It interpolates between positive and completely positive maps. It is shown that this class gives rise to a new characterization of entangled states. Additionally, it provides a refinement of the well-known classes of entangled states characterized in term of the Schmidt number. The analysis is illustrated with examples of qubit maps.
\end{abstract}

\flushbottom

\maketitle

\thispagestyle{empty}

\section{Introduction}

Both positive and completely positive maps play an essential role in quantum information theory \cite{Nielsen}. Recall that a linear map $\Phi : \BH \to \BH$ is positive if $\Phi[X] \geq 0$ for $X  \geq 0$ \cite{Stormer,Stormer2,Paulsen}. Moreover, if $\Phi$ is trace-preserving, then it maps quantum states (represented by density operators) into quantum states. In what follows, we consider only finite-dimensional Hilbert spaces $\HH$, where ${\rm dim}\, \HH =d$. Also, we denote a vector space of (bounded) operators acting on $\HH$ by $\BH$. Interestingly, quantum physics requires a more refined notion of positivity due to the fact that a tensor product $\Phi_1 \otimes \Phi_2$ of two positive maps is not necessarily a positive map. A map $\Phi$ is called {\it $k$-positive} if
\begin{equation}\label{}
{\rm id}_k \otimes \Phi : M_k(\mathbb{C}) \otimes \BH \to M_k(\mathbb{C}) \otimes \BH ,
\end{equation}
is positive, where ${\rm id}_k$ denotes the identity map on $M_k(\mathbb{C})$, which is a vector space of $k \times k$ complex matrices. Finally, a map is completely positive if it is $k$-positive for all $k=1,2,\ldots$. Actually, in the finite-dimensional case, complete positivity is equivalent to $d$-positivity. Hence, if $\mathcal{P}_k$ denotes a (convex) set of $k$-positive, trace-preserving maps, then there is the following chain of inclusions,
\begin{equation}\label{PPP}
  {\rm CPTP \ maps} = \mathcal{P}_d \subset \mathcal{P}_{d-1} \subset \ldots \subset \mathcal{P}_1 = {\rm PTP \ maps} .
\end{equation}
Completely positive maps play a key role in quantum information theory since they correspond to physical operations. In particular, completely positive, trace-preserving (CPTP) maps provide mathematical representations of quantum channels. Any CPTP map satisfies the data processing inequality \cite{Nielsen,Wilde,Hayashi}. Namely, for an arbitrary quantum channel $\mathcal{E}$ and any pair of states $\rho,\sigma$, one has \cite{Lindblad}
\begin{equation}\label{DPI}
   D(\rho||\sigma) \geq  D(\mathcal{E}[\rho]||\mathcal{E}[\sigma]),
\end{equation}
where $D(\rho||\sigma)$ is the relative entropy. Interestingly, it turns out that condition (\ref{DPI}) holds for any PTP map.

Maps that are positive but not completely positive find elegant applications in the theory of entanglement \cite{HHHH,Guhne,TOPICAL,Zyczkowski}. A state $\rho_{AB}$ in $\HH_A \otimes \HH_B$ is separable if and only if
\begin{equation}\label{ent}
  ({\rm id} \otimes \Phi)[\rho_{AB}] \geq 0
\end{equation}
for all positive maps $\Phi$. Hence, any violation of (\ref{ent}) witnesses entanglement of $\rho_{AB}$.
The key property of any PTP map is its contractivity with respect to the trace norm \cite{Kossak1972B},
\begin{equation}\label{}
  \| \Phi[X] \|_{\rm tr} \leq \|X\|_{\rm tr} ,
\end{equation}
for any $X \in \BH$. This implies that distinguishability between any pair of density operators $\rho$ and $\sigma$, defined by $\| \rho - \sigma\|_{\rm tr}$, cannot increase under the action of a PTP map. Similarly, if $\Phi$ is $k$-positive and trace-preserving, then ${\rm id}_k \otimes \Phi$ is contractive.

In this paper, we introduce a new family of maps such that ${\rm id} \otimes \Phi$ is contractive but only on the subspaces of $\mathcal{B}(\HH \otimes \HH)$ of particular dimensions. We call such maps {\it $k$-partially contractive}, where now $k\in \{1,\ldots,d^2\}$. For $k=1$ and $k=d^2$, one reproduces PTP maps and CPTP maps, respectively. Hence, this new family interpolates between these two important classes.
The inspiration for $k$-partially contractive maps comes from \cite{SagnikPQ}, where the author considered the strength of non-Markovian evolution that lies between P and CP-divisible dynamical maps.
We provide the characterisation of partially contractive maps and illustrate this concept with simple qubit maps. Interestingly, in the qubit case, we find a connection between partially contractive maps and the Schwarz maps. The class of partially contractive maps allows us to introduce a new hierarchy of entangled states in full analogy to the well-known characterization in terms of the Schmidt number \cite{Terhal,HHHH,Guhne,TOPICAL}. In the qubit case, this new characterization interpolates between separable states (Schmidt number = 1) and entangled states (Schmidt number = 2). Hence, it provides a refinement of the Schmidt number classes. A simple illustration of two-qubit isotropic states is discussed. We hope that the new class of partially contractive maps introduced in this paper will also allow for a more refined analysis of entangled states in higher-dimensional quantum systems.

\section{Partially contractive maps}

Denote by $\mathcal{B}_{\rm H}(\HH)$ a real subspace of Hermitian operators in $\BH$. Let us recall the following characterisation of PTP maps \cite{Paulsen,Szarek}.

\begin{Proposition} \label{PRO-I} Assume that $\Phi : \BH \to \BH$ is a map that preserves both trace and Hermiticity. Then, $\Phi$ is positive if and only if
\begin{equation}\label{}
  \| \Phi[X] \|_{\rm tr} \leq \|X\|_{\rm tr} ,
\end{equation}
for all Hermitian operators $X \in \mathcal{B}_{\rm H}(\HH)$.
\end{Proposition}

Let $\{\rho_1,\ldots,\rho_k\}$ be a set of linearly independent density operators in $\BHH$ and denote by $\mathcal{M}(\{\rho_1,\ldots,\rho_k\}) = {\rm span}_\mathbb{R} \{\rho_1,\ldots,\rho_k\}$ a real linear subspace of $\mathcal{B}_{\rm H}(\HH)$.

\begin{Definition}\label{def} A trace-preserving, Hermiticity-preserving map $\Phi : \BH \to \BH$ is called {\it $k$-partially contractive} if, for any set of linearly independent $\{\rho_1,\ldots,\rho_k\}$,
\begin{equation}\label{}
  \| ({\rm id} \otimes \Phi)[X] \|_{\rm tr} \leq \| X \|_{\rm tr}
\end{equation}
for all $X \in \BHH \otimes \mathcal{M}(\{\rho_1,\ldots,\rho_k\})$.
\end{Definition}

\begin{Corollary} One has the following correspondence between partial contractivity criteria and positivity of quantum maps:
\begin{enumerate}
  \item  $\Phi$ is PTP iff it is 1-partially contractive;
  \item $\Phi$ is CPTP iff it is $d^2$-partially contractive.
\end{enumerate}
Hence, $k$-partially contractive, trace-preserving maps are interpolated between PTP and CPTP maps.
\end{Corollary}

Denote by $\mathcal{C}_k$ a set of $k$-partially contractive, trace-preserving maps. It is easy to show that there is an inclusion relation between different $\mathcal{C}_k$.

\begin{Proposition} The set $\mathcal{C}_k$ is a convex subset of $\mathcal{P}_1$ (a set of PTP maps). Moreover,
\begin{equation}\label{CCC}
  {\rm CPTP \ maps} = \mathcal{C}_{d^2} \subset \mathcal{C}_{d^2-1} \subset \ldots \subset \mathcal{C}_2 \subset \mathcal{C}_1 = \mathcal{P}_1 = {\rm PTP \ maps} .
\end{equation}
If $\Phi_1 \in \mathcal{C}_k$ and $\Phi_2 \in \mathcal{C}_\ell$, then the composition $\Phi_1 \circ \Phi_2 \in \mathcal{C}_{\min\{k,\ell\}}$.
\end{Proposition}

Given a set $\{\rho_1,\ldots,\rho_k\}$ and a map $\Phi : \BH \to \BH$, define a map restricted to $\mathcal{M}(\{\rho_1,\ldots,\rho_k)\}$ by
\begin{equation}\label{}
  \Phi_\mathcal{M}[X] = \Phi[X]
\end{equation}
for any $X \in \mathcal{M}(\{\rho_1,\ldots,\rho_k\})$.

\begin{Theorem} Let $\Phi : \BH \to \BH$ be a trace-preserving map. If for any linearly independent set $\{\rho_1,\ldots,\rho_k\}$ the restricted map $\Phi_\mathcal{M}$ can be extended to a CPTP map on $\BH$, then $\Phi$ is $k$-partially contractive.
\end{Theorem}

\begin{proof}
If $\widetilde{\Phi}_\mathcal{M}$ is a CPTP extension of $\Phi_\mathcal{M}$, then one has
\begin{equation}\label{}
  \left\| ({\rm id} \otimes \widetilde{\Phi})[X] \right\|_{\rm tr} \leq \|X \|_{\rm tr}
\end{equation}
for any $\BHH \otimes \BHH$. Hence, if $X \in \mathcal{B}_{\rm H}(\HH) \otimes \mathcal{M}(\{\rho_1,\ldots,\rho_k\})$, then
\begin{equation}\label{}
  \left\| ({\rm id} \otimes {\Phi})[X] \right\|_{\rm tr}   =  \left\| ({\rm id} \otimes \widetilde{\Phi})[X] \right\|_{\rm tr}  \leq \|X \|_{\rm tr},
\end{equation}
which proves $k$-partial contractivity of $\Phi$.
\end{proof}

The problem of finding extensions for positive and completely positive maps is well-studied in mathematical literature. Let us recall a seminal extension theorem of Arveson \cite{Arveson}.

\begin{Theorem} Assume that $\Phi : S \to \BH$ is a CP unital map, where $S$ denotes an operator system in $\BH$ (i.e., $\oper \in S$ and if $X \in S$, then $X^\dagger \in S$). Then, there exists a (not unique) CP unital extension $\widetilde{\Phi} : \BH \to \BH$ of $\Phi$ to $\BH$.
\end{Theorem}

Actually, if $S$ contains strictly positive operatora $X>0$ and $\Phi : S \to \BH$ is CP, then there exists a CP extension $\widetilde{\Phi} : \BH \to \BH$ of $\Phi$ \cite{Teiko}. Another interesting result was provided in \cite{Jencova}

\begin{Proposition} Consider a CP map $\Phi : \mathcal{M} \to \BH$, where $\mathcal{M}$ is spanned by positive operators (e.g. density operators). Then, $\Phi$ can be extended to a CP map $\widetilde{\Phi} : \BH \to \BH$.
\end{Proposition}

However, the above result only guarantees the existence of a CP extension and says nothing about trace-preservation.

\section{Qubit maps}

For $d=2$, we have the following seminal result due to Alberti and Uhlmann  \cite{Alberti}.

\begin{Theorem} Let $\Phi : \mathcal{M}(\{\rho_1,\rho_2\}) \to \BH$ be a trace-preserving contraction. Then, $\Phi$ can be extended to a CPTP map $\widetilde{\Phi} : \BH \to \BH$.
\end{Theorem}
This result was recently generalized in \cite{Ende,BuscemiQ}. Using the Albert-Uhlmann theorem, one comes to the following conclusion.

\begin{Corollary} If $d=2$, then any PTP map ${\Phi} : \BH \to \BH$ is a 2-partial contraction.
\end{Corollary}

Hence, in the qubit case, the hierarchy in (\ref{CCC}) simplifies to
\begin{equation}\label{CCC2}
  {\rm CPTP \ maps} = \mathcal{C}_{4} \subset \mathcal{C}_{3}  \subset \mathcal{C}_2 = \mathcal{C}_1 = \mathcal{P}_1 = {\rm PTP \ maps}.
\end{equation}
Therefore, there exists a class $\mathcal{C}_3$ that interpolates between PTP and CPTP maps. In this section, we analyze $\Phi \in \mathcal{C}_3$.

Consider a triple $\{\rho_1,\rho_2,\rho_3\}$ of linearly independent qubit density operators. From \cite{Sagnik}, we know that
\begin{equation}\label{qubit-3d-basis}
  \mathcal{M}(\{\rho_1,\rho_2,\rho_3\}) = {\rm span}_\mathbb{R} \{UX_1U^\dagger,UX_2U^\dagger,UX_3U^\dagger\},
\end{equation}
where $X_1=\sigma_1$, $X_2=\sigma_2$, $X_3=\mathrm{diag}(p,1-p)$ for some $p \in(0,1)$, and $U$ is a unitary operator depending on the choice of $\{\rho_1,\rho_2,\rho_3\}$.
From Definition \ref{def}, $\Phi$ is 3-partially contractive if and only if
\begin{equation}\label{A}
\left\|\sum_{k=1}^3 A_k \otimes \Phi[\rho_k]\right\|_{\rm tr}\leq
\left\|\sum_{k=1}^3 A_k \otimes \rho_k\right\|_{\rm tr}
\end{equation}
for all possible choices of $\{\rho_1,\rho_2,\rho_3\}$ and Hermitian qubit operators $\{A_1,A_2,A_3\}$.
Using eq. (\ref{qubit-3d-basis}), we see that for any set of $\{A_1,A_2,A_3\}$ there exists another set of Hermitian operators $\{B_1,B_2,B_3\}$ such that
\begin{equation}\label{AB}
\sum_{k=1}^3 A_k \otimes \rho_k=\sum_{k=1}^3 B_k \otimes UX_kU^{\dagger}.
\end{equation}
Now, consider a special class of qubit maps with the following property: for any unitary operator $U$, there exists a unitary operator $V$ such that
\begin{equation}\label{cov}
\Phi[UXU^\dagger] =  V\Phi[X]V^\dagger
\end{equation}
for any $X\in\BH$.

\begin{Lemma}\label{3S} If a trace-preserving qubit map $\Phi$ satisfies eq. (\ref{cov}), as well as
\begin{equation}\label{B}
 \left\| \sum_{k=1}^3 B_k \otimes \Phi[X_k] \right\|_{\rm tr}\leq \left\| \sum_{k=1}^3 B_k \otimes X_k \right\|_{\rm tr}
\end{equation}
for any $p\in (0,1)$ and all Hermitian $\{B_1,B_2,B_3\}$, then $\Phi \in \mathcal{C}_3$.
\end{Lemma}

\begin{proof}
It is enough to show that conditions (\ref{A}) and (\ref{B}) are equivalent if $\Phi$ satisfies eq. (\ref{cov}).
Using eq. (\ref{AB}), one has
\begin{equation}
\left\|\sum_{k=1}^3 A_k \otimes \rho_k\right\|_{\rm tr}=\left\|\sum_{k=1}^3 B_k \otimes X_k\right\|_{\rm tr}.
\end{equation}
Now, for a positive, trace-preserving map $\Phi$ satisfying eq. (\ref{cov}), it follows that
\begin{equation}
\left\|\sum_{k=1}^3A_k\otimes\Phi[\rho_k]\right\|_{\rm tr}
=\left\|\sum_{k=1}^3B_k\otimes \Phi[UX_kU^\dagger]\right\|_{\rm tr}
=\left\|\sum_{k=1}^3B_k\otimes V\Phi[X_k]V^\dagger\right\|_{\rm tr}
=\left\|\sum_{k=1}^3B_k\otimes \Phi[X_k]\right\|_{\rm tr}.
\end{equation}
\end{proof}

\begin{Example} Let us propose an example of a positive map that is not 3-partially contractive. Consider the transposition map
\begin{equation}
T[X]=X^T.
\end{equation}
It can be easily seen that $T$ satisfies condition (\ref{cov}). Now, take a set of three Hermitian operators,
\begin{equation}
B_1=\begin{bmatrix}
0 & 1 \\
1 & 0
\end{bmatrix}, \,\,\,
B_2=\begin{bmatrix}
0 & i \\
-i & 0
\end{bmatrix}, \,\,\,
B_3=\begin{bmatrix}
1 & 0 \\
0 & 0
\end{bmatrix}.
\end{equation}
The corresponding trace norms read
\begin{equation}
\left\|\sum_{k=1}^3 B_k \otimes X_k \right\|_{\rm tr}=\begin{Vmatrix}
p & 0 & 0 & 0 \\
0 & 1-p & 2 & 0 \\
0 & 2 & 0 & 0\\
0 & 0 & 0 & 0
\end{Vmatrix}_{tr}=p+\sqrt{(1-p)^2+16}
\end{equation}
and
\begin{equation}
\left\|\sum_{k=1}^3 B_k \otimes T[X_k] \right\|_{\rm tr}=\begin{Vmatrix}
p & 0 & 0 & 2 \\
0 & 1-p & 0 & 0 \\
0 & 0 & 0 & 0\\
2 & 0 & 0 & 0
\end{Vmatrix}_{tr}=1-p+\sqrt{p^2+16}.
\end{equation}
Hence, we show that $T$ violates condition (\ref{B}) for $ p > \frac 12$ and thus, from Lemma \ref{3S}, we conclude that $T\notin\mathcal{C}_3$.
\end{Example}

\begin{Example} Consider the qubit map
\begin{equation}   \label{Lam}
\Lambda_a[X] = \frac{1}{2-a}\left(\oper {\rm Tr}X - a X\right),
\end{equation}
which is trace-preserving and positive if and only if $0\leq a\leq 1$. Let us show when $\Lambda_a$ is 3-partially contractive but not completely positive. Observe that $0\leq a\leq 1$ gives the primary constraint for $a$, as 3-partially contractive maps are necessarily positive. The map restricted to the subspace $\mathcal{M}$ reads
\begin{equation}
\Lambda_a|_{\mathcal{M}}[X_1]=-\frac{a}{2-a}X_1,\qquad \Lambda_a|_{\mathcal{M}}[X_2]=-\frac{a}{2-a}X_2,\qquad \Lambda_a|_{\mathcal{M}}[X_3]=\frac{1}{2-a}(\oper-a X_3).
\end{equation}
Now, we extend $\Lambda_a|_{\mathcal{M}}$ to $\widetilde{\Lambda}_a:\mathcal{B}(\mathcal{H})\to\mathcal{B}(\mathcal{H})$ in the following way,
\begin{equation}
\widetilde{\Lambda}_a[X_k]=\Lambda_a|_{\mathcal{M}}[X_k],\quad k=1,2,3,\qquad \widetilde{\Lambda}_a[X_4]=r\sigma_3,
\end{equation}
$\sigma_3=\mathrm{diag}(1,-1)$, which guarantees the trace preservation of $\widetilde{\Lambda}_a$.
The complete positivity of the extension is equivalent to the positivity of its Choi matrix \cite{Choi}
\begin{equation}
C(\widetilde{\Lambda}_a)=\frac 12 \sum_{j,k=0}^1|j\>\<k|\otimes\Lambda[|j\>\<k|]
=\frac{1}{2(2-a)}
\begin{pmatrix}
c(p) & 0 & 0 & -a \\
0 & b(1-p) & 0 & 0 \\
0 & 0 & b(p) & 0 \\
-a & 0 & 0 & c(1-p)
\end{pmatrix},
\end{equation}
where
\begin{equation}
c(p)=1-pa+r(1-p)(2-a),\qquad b(p)=1-pa-rp(2-a).
\end{equation}
From Sylvester's criterion \cite{Gilbert}, $C(\widetilde{\Lambda}_a)\geq 0$ if and only if all its minors have positive determinants. This translates to the condition that there exists a real number $r$ such that
\begin{equation}\label{21}
c(p)\geq 0,\qquad b(p)\geq 0,\qquad \det A(p)=
\det\begin{pmatrix}
c(p) & -a \\
-a & c(1-p)
\end{pmatrix}\geq 0
\end{equation}
for all $0\leq p\leq 1$. From the first two inequalities, one gets
\begin{equation}
-\frac{1}{2-a}\leq r\leq \frac{1-a}{2-a}.
\end{equation}
What remains is the condition for the determinant of $A(p)$.
For the allowed range of $p$, the only local extrema of $\det A(p)$ are maxima, as one has
\begin{equation}
\frac{\partial^2\det A(p)}{\partial p^2}=-2(ra-a-2r)^2<0.
\end{equation}
Therefore, the minimal value of $\det A(p)$ is achieved at the end points, $p=0$ and $p=1$. The sufficient condition for the positivity of $\det A(p)$ is the positivity of $\det A(0)$, where
\begin{equation}
\det A(0)=\det A(1)=(1-a)[1+r(2-a)]-a^2.
\end{equation}
The last inequality of eq. (\ref{21}) gives the new upper bound on $r$,
\begin{equation}
r\geq\frac{a^2+a-1}{a^2-3a+2}.
\end{equation}
Thus, there exists a parameter $r$ such that $\widetilde{\Lambda}_a[X_3]=r\sigma_3$ whenever
\begin{equation}
\frac{a^2+a-1}{a^2-3a+2}\leq\frac{1-a}{2-a},
\end{equation}
which is satisfied for $a\in[0,1]$ if and only if
\begin{equation}
0\leq a\leq \frac 23.
\end{equation}
This way, we arrive at the 3-partial contractivity condition for $\Lambda_a$. Finally, $\Lambda_a$ is completely positive if and only if $0\leq a\leq 1/2$. Hence, it is 3-partially contractive but not completely positive for
\begin{equation}
\frac 12<a\leq \frac 23.
\end{equation}
\end{Example}

\begin{Example}
Now, consider another trace-preserving qubit map,
\begin{equation}   \label{Om}
\Omega_\epsilon[X] = \frac{\epsilon}{2} \oper {\rm Tr}X +(1-\epsilon) X^T.
\end{equation}
This map is positive if and only if $0\leq\epsilon\leq 1$ and completely positive if and only if $2/3\leq\epsilon\leq 1$. Using the same method as in the previous example, one shows that $\Omega_\epsilon$ is 3-partially contractive but not completely positive for
\begin{equation}
\frac 12\leq\epsilon<\frac 23.
\end{equation}
\end{Example}

\section{Partial contractivity vs. Schwarz qubit maps}

A positive map $\Phi : \BH \to \BH$ is called the {\it Schwarz map} if for any $X \in \BH$,
\begin{equation}\label{Smap}
  \| \Phi(\oper)\|_\infty \Phi[X^\dagger X]\geq \Phi[X^\dagger]\Phi[X] ,
\end{equation} 
where $\| A \|_\infty$ stands for the operator norm. Any Schwarz map satisfies $\|\Phi\|_\infty =  \| \Phi(\oper)\|_\infty$, where
\begin{equation}\label{}
  \|\Phi\|_\infty := \sup_{X \in \BH} \frac{\| \Phi(X)\|_\infty}{\| X \|_\infty} .
\end{equation}
If $\Phi$ is unital ($\Phi[\oper]=\oper$), then eq. (\ref{Smap}) reduces to \cite{Paulsen,Bhatia,Ringrose}
\begin{equation}\label{KS}
\Phi[X^\dagger X]\geq \Phi[X^\dagger]\Phi[X] .
\end{equation}
Note that condition (\ref{Smap}) is sufficient for positivity and necessary for complete positivity. A composition  $\Phi_1\circ\Phi_2$ and a convex combination $q\Phi_1+(1-q)\Phi_2$ of two unital Schwarz maps is also a unital Schwarz map. It was proved by Kadison that any positive unital map satisfies eq. (\ref{KS}) for Hermitian $X$ (the celebrated Kadison inequality \cite{Kadison}).  If $\Phi$ is not unital but $V=\Phi[\oper]>0$, then $\Psi(X) := V^{-1/2} \Phi[X] V^{-1/2}$ is inital and $\Psi$ is a Schwarz map if and only if
\begin{equation}
\Phi[X^\dagger X]\geq \Phi[X^\dagger]\Phi[\oper]^{-1}\Phi[X] \geq \frac{1}{ \|\Phi\|_\infty} \Phi[X^\dagger]\Phi[X] ,
\end{equation}
for arbitrary $X\in\mathcal{B}(\mathcal{H})$. Hence, $\Phi$ satisfies (\ref{Smap}). Now, for PTP unital qubit maps, one finds the following hierarchy,
\begin{equation}\label{relationKS}
{\rm CPTP\ unital\ maps} \subset S \subset {\rm PTP\ unital\ maps} ,
\end{equation}
where $S$ denotes the set of Schwarz maps. After comparing eq. (\ref{CCC2}) with eq. (\ref{relationKS}), it would be interesting to analyze the relation between two classes of maps: 3-partially contractive and Schwarz maps.

A method of constructing the Schwarz maps was proposed in \cite{KSApprox}. Take a positive, trace-preserving, unital map $\Psi:\mathcal{B}(\mathcal{H})\to\mathcal{B}(\mathcal{H})$ that is a contraction  $\|\Psi[X]\|_2\leq\|X\|_2$ in the Hilbert-Schmidt (Frobenius) norm,
where $\|X\|_2=\sqrt{\Tr (X^\dagger X)}$. Define the qubit map
\begin{equation}
\Phi[X]=\frac{q}{2}\oper\Tr X+(1-q)\Psi[X].
\end{equation}
Now, $\Phi$ is a Schwarz map if
\begin{equation}
\frac 12 \leq q \leq \frac 32.
\end{equation}
The following stronger results were also proved.

\begin{Proposition}[\cite{KSApprox}]
The map $\Lambda_a$ defined in eq. (\ref{Lam}) is the Schwarz map (that is not CP) if and only if $\frac 12< a \leq \frac 23$.
\end{Proposition}

\begin{Proposition}[\cite{KSApprox}]
The map $\Omega_\epsilon$ defined in eq. (\ref{Om}) is the Schwarz map (that is not CP) if and only if $\frac 12\leq\epsilon<\frac 23$.
\end{Proposition}

At this point, we make an important observation: the sufficient condition for 3-partial contractivity and the necessary and sufficient conditions for Schwarz maps coincide for $\Lambda_\epsilon$ and $\Omega_\epsilon$. In other words, the following relations hold,
\begin{equation}
{\rm CPTP\ unital\ maps} \subset \mathcal{C}_3\subset S \subset {\rm PTP\ unital\ maps}.
\end{equation}
Therefore, one observes an intricate connection between the Schwarz maps and 3-partially contractive maps. This problem deserves further analysis.

\section{Partial contractions vs. quantum entanglement}

Any state vector $\psi \in \HH \otimes \HH$ gives rise to the Schmidt decomposition
\begin{equation}\label{}
  |\psi\> = \sum_{i=1}^r s_i |e_i\> \otimes |f_i\> ,
\end{equation}
where the Schmidt rank $r={\rm SR}(\psi)$ satisfies $1 \leq r \leq d$.  This concept can be easily generalized to density operators \cite{Terhal}: given $\rho$, one defines its Schmidt number
\begin{equation}\label{}
  {\rm SN}(\rho) = \min_{\{p_k,\psi_k\}} \max_k {\rm SR}(\psi_k),
\end{equation}
where the minimization is carried over all pure state decompositions $\rho = \sum_k p_k |\psi_k\>\<\psi_k|$. If $\rho = |\psi\>\<\psi|$, then ${\rm SN}(\rho) = {\rm SR}(\psi)$.

\begin{Proposition}[\cite{TOPICAL}] A map $\Phi : \BH \to \BH$ is $k$-positive if and only if
\begin{equation}\label{}
   ({\rm id} \otimes \Phi)[\rho] \geq 0
\end{equation}
for any $\rho \in \mathcal{B}(\HH \otimes \HH)$ with ${\rm SN}(\rho) \leq k$.
\end{Proposition}

This results allows us to provide the following classification of entangled states in $\HH \otimes \HH$,
\begin{equation}\label{EEE}
  {\rm separabale \ states} = \mathbb{E}_1 \subset \mathbb{E}_2 \subset \ldots \subset \mathbb{E}_{d-1} \subset \mathbb{E}_d = {\rm all \ states} ,
\end{equation}
where $\mathbb{E}_k$ contains all states with ${\rm SN} \leq k$. The hierarchy in (\ref{EEE}) is dual to (\ref{PPP}).

\begin{Example} \label{k-P} Consider a map $\Phi : \BH \to \BH$ defined by
\begin{equation}\label{Phi-k}
  \Phi_p[X] = p\, \oper_d {\rm Tr} X - X ,
\end{equation}
with $p > 0$. Choi showed that $\Phi_p$ is $k$-positive but not $(k+1)$-positive if and only if $k \leq p < k+1$ \cite{Choi-72}. In particular, $\Phi$ is positive if $p \geq 1$. In the entanglement theory, the map $\Phi_1$ is called the {\it reduction map} and plays an important role in classifying states of composite systems \cite{HHHH,Guhne,TOPICAL}. Now, consider the family of isotropic states in $\HH \otimes \HH$,
\begin{equation}\label{iso}
  \rho_f = \frac{1-f}{d^2-1} (\oper_d \otimes \oper_d - P^+_d) + f P^+_d , \ \ \ \ \ f\in [0,1] ,
\end{equation}
where $P^+_d$ denotes the projector onto the maximally entangled state. One finds that if the fidelity $f > k/d$, then ${\rm SN}(\rho) \geq k+1$ \cite{Terhal}.
\end{Example}
A large class of $k$-positive maps based on spectral property of the Choi matrix was proposed in \cite{DCAK}. It provides a generalization of the Choi map from Example \ref{k-P}.

Note that one can easily define a hierarchy similar to (\ref{CCC}) via
\begin{equation}\label{EEE-ent}
  {\rm separabale \ states} = \mathcal{E}_1 \subset \mathcal{E}_2 \subset \ldots \subset \mathcal{E}_{d^2-1} \subset \mathcal{E}_{d^2} = {\rm all \ states} ,
\end{equation}
where $\mathcal{E}_k$ contains such states $\rho$ for which
\begin{equation}\label{}
  ({\rm id} \otimes \Phi)[\rho] \geq 0
\end{equation}
for all $\Phi \in \mathcal{C}_k$. In the qubit case, the above hierarchy reduces to
\begin{equation}\label{EEE-ent2}
  {\rm separabale \ states} = \mathcal{E}_1 = \mathcal{E}_2 \subset  \mathcal{E}_{3} \subset \mathcal{E}_{4} = {\rm all \ states} .
\end{equation}

\begin{Example}
Consider the action of $\Lambda_a$ from eq. (\ref{Lam}) on one part of the two-qubit isotropic states $\rho_f$ defined in eq. (\ref{iso}).
We find that
\begin{equation}
({\rm id}\otimes \Lambda_a)[\rho_f]=
\frac{1}{6 (2-a)}\begin{bmatrix}
3-a(1+2f) & 0 & 0 & -a (4f-1)  \\
0 & 3-2a(1-f) & 0 & 0 \\
0 & 0 & 3-2a(1-f) & 0 \\
-a (4f-1) & 0 & 0 & 3-a(1+2f)
\end{bmatrix}.
\end{equation}
The above matrix is positive in the range $\frac{1}{2} < a \leq \frac{2}{3}$, where $\Lambda_a$ is 3-partially contractive but not completely positive, if and only if
\begin{equation}
0\leq f\leq \frac{3}{4}.
\end{equation}
The same result is obtained if one replaces considers $\Omega_\epsilon$ instead of $\Lambda_a$.
Thus, the above inequality provides a necessary condition for $\rho_f\in\mathcal{E}_3$.
Note that $\rho_f\notin\mathcal{E}_2$ (is entangled) for $1/2<f\leq 1$, and if $f>3/4$, then $\rho_f\notin\mathcal{E}_3$ ({\em highly entangled}).
\end{Example}

\section{Conclusions}

We propose a new classification of positive maps based on the contractivity with respect to the trace norm on subspaces of $\mathcal{B}(\mathcal{H}\otimes\mathcal{H})$. The $k$-partially contractive maps give rise to a hierarchy of $d^2$ classes of positive maps that interpolate between PTP and CPTP maps. The parameter $k$, which is the dimension of the subspace, can be interpreted as a strength of positivity. For qubit maps, we introduce an intermediate class of positive but not completely positive maps, corresponding to the choice $k=3$. Moreover, we provide an analytical technique for assessing 3-partial contractivity. Our concept is illustrated with examples of simple qubit maps. Interestingly, our analysis shows that there is a connection between partially contractive maps and the Schwarz maps. Finally, we apply these results to refine the hierarchy of entangled states based on the Schmidt number. In the qubit case, we obtain a single class that interpolates between separable and entangled states.

The topic of partially contractive maps requires further analysis. The first step would be to find less restrictive sufficient conditions for 3-partially contractive qubit maps. Then, the full relation between these maps and the Schwarz maps could be established. Also, it is crucial to find a computational method of constructing $k$-partially contractive maps in any finite dimension $d$. To achieve this, one needs a relation analogical to eq. (\ref{qubit-3d-basis}) for $d=2$. Then, one would be able to find the connection between $k$-partial contractivity and $k$-positivity. Another open question concerns possible applications of our classification. One implementation, already touched upon in this paper, is a more refined analysis of entangled states.

\section{Acknowledgements}

This paper was supported by the Polish National Science Centre project No. 2018/30/A/ST2/00837.

\bibliography{C:/Users/cynda/OneDrive/Fizyka/bibliography}

\begin{thebibliography}{10}
\providecommand{\url}[1]{\texttt{#1}}
\providecommand{\urlprefix}{URL }
\providecommand{\eprint}[2][]{\url{#2}}

\bibitem{Nielsen}
M.~A. Nielsen and I.~L. Chuang, \textit{Quantum Computation and Quantum
  Information},  Cambridge University Press, Cambridge 2010.

\bibitem{Stormer}
E.~St{\o}rmer, \textit{Positive linear maps of operator algebras},
  Springer-Verlag, Berlin 2013.

\bibitem{Stormer2}
E.~St{\o}rmer, Acta Math. \textbf{110}, 233--278 (1963).

\bibitem{Paulsen}
V.~Paulsen, \textit{Completely Bounded Maps and Operator Algebras},  Cambridge
  University Press, Cambridge 2003.

\bibitem{Wilde}
M.~M. Wilde, \textit{Quantum Information Theory},  Cambridge University Press,
  Cambridge 2013.

\bibitem{Hayashi}
M.~Hayashi, \textit{Quantum Information: An Introduction},  Springer, Berlin
  2006.

\bibitem{Lindblad}
G.~Lindblad, Comm. Math. Phys. \textbf{40}, 147 (1975).

\bibitem{HHHH}
R.~Horodecki, P.~Horodecki, M.~Horodecki, and K.~Horodecki, Rev.Mod. Phys.
  \textbf{81}, 865 (2009).

\bibitem{Guhne}
O.~G{\"u}hne and G.~T{\'o}th, Phys. Rep. \textbf{474}, 1 (2009).

\bibitem{TOPICAL}
D.~Chru\'sci\'nski and G.~Sarbicki, J. Phys. A: Math. Theor. \textbf{47},
  483001 (2014).

\bibitem{Zyczkowski}
I.~Bengtsson and K.~\.{Z}yczkowski, \textit{Geometry of Quantum States: An
  Introduction to Quantum Entanglement},  Cambridge University Press, Cambridge
  2007.

\bibitem{Kossak1972B}
A.~Kossakowski, Bull. Acad. Polon. Sci. S{\'e}r. Math. Astr. Phys \textbf{20},
  1021--1025 (1972).

\bibitem{SagnikPQ}
S.~Chakraborty, Phys. Rev. A \textbf{97}, 032130 (2018).

\bibitem{Terhal}
B.~M. Terhal and P.~Horodecki, Phys. Rev. A \textbf{61}, 040301(R) (2000).

\bibitem{Szarek}
M.~B. Ruskai, S.~Szarek, and E.~Werner, Linear Algebra Appl. \textbf{347(1-3)},
  159--187 (2002).

\bibitem{Arveson}
W.~Arveson, Acta Math. \textbf{123}, 141 (1969).

\bibitem{Teiko}
T.~Heinosaari, M.~A. Jivulescu, D.~Reeb, and M.~M. Wolf, J. Math. Phys.
  \textbf{53}, 102208 (2012).

\bibitem{Jencova}
A.~Jen\u{c}ov\'{a}, J. Math. Phys. \textbf{53}, 012201 (2012).

\bibitem{Alberti}
P.~M. Alberti and A.~Uhlmann, Rep. Math. Phys. \textbf{18}, 163--176 (1980).

\bibitem{Ende}
S.~Chakraborty, D.~Chru\'{s}ci\'{n}ski, G.~Sarbicki, and F.~vom Ende, Quantum
  \textbf{4}, 360 (2020).

\bibitem{BuscemiQ}
M.~Dall'Arno, F.~Buscemi, and V.~Scarani, Quantum \textbf{4}, 233 (2020).

\bibitem{Sagnik}
S.~Chakraborty and D.~Chru\'{s}ci\'{n}ski, Phys. Rev. A \textbf{99}, 042105
  (2019).

\bibitem{Choi}
M.-D. Choi, Linear Algebra Appl. \textbf{10}, 285--290 (1975).

\bibitem{Gilbert}
G.~Gilbert, Am. Math. Mon. \textbf{98}, 44--46 (1991).

\bibitem{Bhatia}
R.~Bhatia, \textit{Positive Definite Matrices},  Princeton University Press,
  Princeton 2006.

\bibitem{Ringrose}
R.~V. Kadison and J.~Ringrose, \textit{Fundamentals of the Theory of Operator
  Algebras},  Academic Press, New York 1956.

\bibitem{Kadison}
R.~V. Kadison, Ann. of Math. \textbf{56}, 494--503 (1952).

\bibitem{KSApprox}
D.~Chru\'{s}ci\'{n}ski, F.~Mukhamedov, and M.~A. Hajji, Open Sys. Inf. Dyn.
  \textbf{27}, 2050016 (2020).

\bibitem{Choi-72}
M.-D. Choi, Canad. J. Math. \textbf{24}, 520--529 (1972).

\bibitem{DCAK}
D.~Chru\'{s}ci\'{n}ski and A.~Kossakowski, Commun. Math. Phys. \textbf{290},
  1051 (2009).

\end{thebibliography}
\bibliographystyle{C:/Users/cynda/OneDrive/Fizyka/beztytulow2}

\end{document}